\documentclass[11pt]{article}
\usepackage{url}
\usepackage{fullpage}
\usepackage{./jeremy}
\usepackage{./preorderMeasureNotations}
\usepackage{./succinctEncodingsNotations}
\usepackage{amssymb}

\newtheorem{theorem}{Theorem}
\newtheorem{definition}{Definition}
\newtheorem{corollary}[theorem]{Corollary}

\newenvironment{proof}{\noindent {\bf\em Proof.}}{\hfill $\Box$ \\ }

\title{On Compressing Permutations and Adaptive Sorting
\thanks{Partially funded by Fondecyt Grant 1-110066, Chile.
An early partial version of this paper appeared in {\em STACS} 
\cite{BN09}.}}

\author{%
\begin{tabular}{cc}
J\'er\'emy Barbay & Gonzalo Navarro \\ 
\end{tabular}
\ \\
\ \\
\small Dept. of Computer Science, University of Chile}

\date{}

\begin{document}

\maketitle

\begin{abstract}
  Previous compact representations of permutations have focused on
  adding a small index on top of the plain data $\langle \pi(1), \pi(2),
  \ldots \pi(n) \rangle$, in order to efficiently support the
  application of the inverse or the iterated permutation.
  In this paper we initiate the study of techniques that exploit the
  compressibility of the data itself, while retaining efficient
  computation of $\pi(i)$ and its inverse.
  In particular, we focus on exploiting {\em runs}, which are subsets
  (contiguous or not) of the domain where the permutation is monotonic.
  Several variants of those types of runs arise in real applications
  such as inverted indexes and suffix arrays.
  Furthermore, our improved results on compressed data structures for
  permutations also yield better adaptive sorting algorithms.
\end{abstract}

\section{Introduction}
\label{sec:introduction}

Permutations of the integers $[1..n] = \{1,\ldots,n\}$ are not only a fundamental
mathematical structure, but also a
basic building block for the succinct encoding of integer
functions~\cite{representingFunctions},
strings~\cite{Kar99,rankSelectOperationsOnLargeAlphabets,GV06,ANS06,MN07,CHSV08},
binary
relations~\cite{succinctIndexesForStringsBinaryRelationsAndMultiLabeledTrees},
and geometric grids~\cite{BLNS09}, among others.
A permutation $\pi$ can be trivially encoded in $n\lceil\lg n\rceil$
bits, which is within $\bigo(n)$ bits of the information theory lower
bound of $\lg(n!)$ bits, where $\lg x=\log_2 x$ denotes the logarithm
in base two.

In most of those applications, efficient
computation is required for both the value $\pi(i)$ at any point
$i\in[1..n]$ of the permutation, and for the position $\pi^{-1}(j)$ of
any value $j\in[1..n]$ (i.e., the value of the inverse permutation).
The only alternative we are aware of to storing explicitly both $\pi$ and
$\pi^{-1}$ is by Munro et al.~\cite{succinctRepresentationsOfPermutations},
who add a small structure over the plain representation of $\pi$ so that,
by spending $\epsilon \lg n$ extra bits, any $\pi^{-1}(j)$ can be computed
in time $\bigo(1/\epsilon)$. This is extended to any positive or negative
power of $\pi$, $\pi^k(i)$. They give another solution using $\bigo(n)$
extra bits and computing any $\pi^k(j)$ in time $\bigo(\lg n / \lg\lg n)$.

The lower bound of $\lg(n!)$ bits yields a lower bound of $\Omega(n\lg
n)$ comparisons to sort such a permutation in the comparison model, in
the worst case over all permutations of $n$ elements.
Yet, a large body of research has been dedicated to finding better
sorting algorithms which can take advantage of specificities of each
permutation to sort.
Some examples are permutations composed of a few sorted
blocks~\cite{measuresOfPresortednessAndOptimalSortingAlgorithms}
(e.g., $(\mathit{1},\mathit{3},\mathit{5},\mathit{7},\mathit{9},\mathbf{2},\mathbf{4},\mathbf{6},\mathbf{8},\mathbf{10})$
or
$(\mathit{6},\mathit{7},\mathit{8},\mathit{9},\mathit{10},\mathbf{1},\mathbf{2},\mathbf{3},\mathbf{4},\mathbf{5})$),
or permutations containing few sorted
subsequences~\cite{sortingShuffledMonotoneSequences}
(e.g., $(\mathit{1},\mathbf{6},\mathit{2},\mathbf{7},\mathit{3},\mathbf{8},\mathit{4},\mathbf{9},\mathit{5},\mathbf{10})$).
Algorithms performing possibly $o(n\lg n)$ comparisons on such
permutations, yet still $\bigo(n\lg n)$ comparisons in the worst case,
are achievable and preferable if those permutations arise with
sufficient frequency.
Other examples are classes of permutations whose structure makes them
interesting for applications: see the seminal paper of 
Mannila~\cite{measuresOfPresortednessAndOptimalSortingAlgorithms}, and
the survey of Moffat and Petersson~\cite{anOverviewOfAdaptiveSorting} for
more details.

Each sorting algorithm in the comparison model yields an encoding
scheme for permutations: the result of all comparisons performed
uniquely identifies the permutation sorted, and hence encodes it.
Since an adaptive sorting algorithm performs $o(n\lg n)$ comparisons
on a class of ``easy'' permutations, each adaptive algorithm yields a
{\em compression scheme} for permutations, at the cost of losing a
constant factor on the complementary class of ``hard'' permutations.
Yet such compression schemes do not necessarily support efficiently
the computation of value $\pi^{-1}(j)$ of the inverse permutation for
an arbitrary value $j\in[1..n]$, or even the simple application of the
permutation, $\pi(i)$. 

This is the topic of our study: the interplay between adaptive sorting
algorithms and compressed representation of permutations that support
efficient application of $\pi(i)$ and $\pi^{-1}(j)$. In particular we
focus on classes of permutations that can be decomposed into a small
number of {\em runs}, that is, monotone subsequences of $\pi$, either
contiguous or not. 

Our results include compressed representations of permutations whose space
and time to compute any $\pi(i)$ and $\pi^{-1}(j)$ are proportional to the
{\em entropy} of the distribution of the sizes of the runs. As far as we 
know, this is the first compressed representation of permutations with similar
capabilities.

We also develop the corresponding sorting algorithms, which in general refine 
the known complexities to sort those classes of permutations: While there
exist sorting algorithms taking advantage of the number of runs of various
kinds, ours take advantage of their size distribution and are strictly
better (or equal, at worst). 

Finally, we
obtain a representation for strings that improves upon the state of the
art \cite{FMMN07,GRR08} in the average case, while retaining
their space and worst-case performance for operations access, rank, and
select.

At the end of the article we describe some applications where the class
of permutations compressible with the techniques we develop here
naturally arise, and conclude with a more general perspective 
on the meaning of those results and the research directions they suggest.

\section{Basic Concepts and Previous Work}
\label{sec:previous-work}

\subsection{Entropy}
\label{sec:compr-subs}

We define the {\em entropy} of a distribution \cite{CT91}, a measure that 
will be useful to evaluate compressibility results. 

\begin{definition}
  The {\em entropy} of a sequence of positive integers $X=\langle n_1,n_2,
  \ldots, n_r\rangle$ adding up to $n$ is $\entropy(X) =
  \sum_{i=1}^r\frac{n_i}{n} \lg \frac{n}{n_i}$.
  By concavity of the logarithm, it holds that
    $(r-1)\lg n \le n\entropy(X) \le n\lg r$
   and that $\entropy(\langle n_1,n_2,\ldots n_r\rangle) >
	     \entropy(\langle n_1{+}n_2,\ldots,n_r\rangle)$.
\label{def:entrop}
\end{definition}

Here $\langle n_1,n_2, \ldots, n_r\rangle$ is a distribution of values
adding up to $n$ and $\entropy(X)$ measures how even is the distribution.
$\entropy(X)$ is maximal ($\lg r$) when all $n_i = n/r$ and minimal 
($\frac{r-1}{n}\lg n + \frac{n-r+1}{n}\lg\frac{n}{n-r+1}$) when they are
most skewed ($X=\langle 1,1,\ldots,1,n-r+1\rangle$). 

This measure is related to entropy of random variables and of sequences as
follows.
If a random variable $P$ takes the value $i$ with probability $n_i/n$, for
$1\le i\le r$, then its entropy is $\entropy(\langle n_1,n_2, \ldots, 
n_r\rangle)$. Similarly, if a string $S[1..n]$ contains $n_i$ occurrences of 
character $c_i$, then its empirical zero-order entropy is $\entropy_0(S) =
\entropy(\langle n_1,n_2, \ldots, n_r\rangle)$.

$\entropy(X)$ is then a lower bound to the average number of bits needed to
encode an instance of $P$, or to encode a character of $S$ (if we model $S$
statistically with a zero-order model, that is, ignoring the context of
characters). 

\subsection{Huffman Coding}

The Huffman algorithm~\cite{Huf52} receives frequencies
$X = \langle n_1, n_2, \ldots, n_r \rangle$ adding up to $n$,
and outputs in $\bigo(r \lg r)$ time a prefix-free code for the
symbols $[1..r]$. If $\ell_i$ is the bit length of the code assigned
to the $i$th symbol, then $L=\sum \ell_i n_i$ is minimal.
Moreover, $L < n(1+\entropy(X))$. For example, given
$S[1..n]$ over alphabet $[1..r]$, with symbol frequencies $X$,
one can compress $S$ by concatenating the codewords of the successive
symbols $S[i]$, achieving total length $L < n(1+\entropy_0(S))$. (One
also has to encode the usually negligible codebook of $\bigo(r \lg r)$ bits.)

Huffman's algorithm starts with a forest of $r$ leaves
corresponding to the frequencies $\{ n_1,n_2, \ldots,n_r\}$,
and outputs a binary trie with those leaves, in some order.
This so-called {\em Huffman tree} describes the optimal encoding as
follows:
The sequence of left/right choices (interpreted as 0/1)
in the path from the root to each
leaf is the prefix-free encoding of that leaf, of length $\ell_i$
equal to the leaf depth.

A generalization of this encoding is multiary Huffman coding \cite{Huf52},
in which the tree is given arity $t$, and then the Huffman codewords are 
sequences over an alphabet $[1..t]$. In this case the algorithm also produces 
the optimal code, of length $L < n(1+\entropy(X)/\lg t)$. 

\subsection{Succinct Data Structures for Sequences}
\label{sec:sequences}

Let $S[1..n]$ be a sequence of symbols from the alphabet $[1..r]$.
This includes bitmaps when $r=2$ (where, for convenience, the alphabet
will be $\{0,1\}$ rather than $\{1,2\}$). 
We will make use of succinct representations of $S$ that support the
rank and select operators over strings 
and over binary vectors: 
$\StrRank_c(S,i)$ gives the number of occurrences of $c$ in $S[1..i]$
and $\StrSelect_c(S,j)$ gives the position in $S$ of the $j$th
occurrence of $c$.

When $r=2$, $S$ requires $n$ bits and
$\BinRank$ and $\BinSelect$ can be supported in constant time using
$\bigo(n\llg n/\lg n) = o(n)$ bits on top of $S$ \cite{Mun96,Gol06}. 

Raman et al.~\cite{RRR02} devised a bitmap representation that takes
$n\entropy_0(S) + o(n)$ bits, while maintaining the constant time for
supporting the operators. For the binary case $\entropy_0(S)$ is just
$m\lg\frac{n}{m} + (n-m)\lg\frac{n}{n-m} = m\lg\frac{n}{m}+\bigo(m)$,
where $m$ is the number of bits set to $1$ in $S$. 
Golynski et al.~\cite{GGGRR07} reduced the $o(n)$-bits redundancy 
in space to $\bigo(n\lg\lg n/\lg^2 n)$.

When $m$ is much smaller than $n$, the $o(n)$-bits term may dominate.
Gupta et al.~\cite{GHSV06} showed how to achieve space $m\lg\frac{n}{m} +
\bigo(m\llg\frac{n}{m}+\lg n)$ bits, which largely reduces the dependence
on $n$, but now $\BinRank$ and $\BinSelect$ are supported in
$\bigo(\lg m)$ time via binary search \cite[Theorem 17 p.~153]{Gup07}.

For larger alphabets, of size $r = \bigo(\textrm{polylog}(n))$, 
Ferragina et al.~\cite{FMMN07} showed how to represent the sequence 
within $n\entropy_0(S) + o(n\lg r)$ bits and support $\StrRank$
and $\StrSelect$ in constant time.
Golynski et al.~\cite[Lemma 9]{GRR08} improved the space to $n\entropy_0(S) +
o(n\lg r/\lg n)$ bits while retaining constant times.

Grossi et al.~\cite{GGV03} introduced the so-called {\em wavelet tree}, which 
decomposes an arbitrary sequence into several bitmaps. By representing the 
bitmaps in compressed form \cite{GGGRR07}, the overall space is $n\entropy_0(S) +
o(n)$ and $\StrRank$ and $\StrSelect$ are supported in time
$\bigo(\lg r)$. Multiary wavelet trees decompose the sequence into
subsequences over a sublogarithmic-sized alphabet and reduce the time to
$\bigo(1+\lg r/\llg n)$ \cite{FMMN07,GRR08}.

In this article $n$ will generally denote the length of the permutation.
All of our $o()$ expressions, even those including several variables, will
be asymptotic in $n$.

\subsection{Measures of Presortedness in Permutations}
\label{sec:meas-disord-perm}

The complexity of \emph{adaptive algorithms}, for problems such as
searching, sorting, merging sorted arrays or convex hulls, is
studied in the worst case over instances of fixed size \emph{and
  difficulty}, for a definition of difficulty 
that is specific to each analysis.
Even though sorting a permutation in the comparison model requires
$\Theta(n\lg n)$ comparisons in the worst case over permutations of
$n$ elements, better results can be achieved for some parameterized
classes of permutations. We describe some of those below, see the
survey by Moffat and Petersson~\cite{anOverviewOfAdaptiveSorting}
for others.

Knuth~\cite{theArtOfComputerProgrammingVol3} considered \emph{runs}
(contiguous ascending subsequences) of a permutation $\pi$, counted by
$\nRuns=1+|\{i: 1 \leq i < n, \pi(i+1)<\pi(i)\}|.$
Levcopoulos and Petersson~\cite{sortingShuffledMonotoneSequences}
introduced \emph{Shuffled Up-Sequences} and its generalization
\emph{Shuffled Monotone Sequences}, respectively counted by
$\nSUS=\min\{k:\pi\textrm{ is covered by }k\textrm{ increasing }$
$\textrm{subsequences} \},$ and $\nSMS=\min\{k:\pi\textrm{ is covered
  by }k\textrm{ monotone}$ $\textrm{subsequences} \}$.
By definition, $\nSMS\leq\nSUS\leq\nRuns$.

Munro and Spira~\cite{sortingAndSearchingInMultisets} took an
orthogonal approach, considering the task of sorting multisets through
various algorithms such as MergeSort, showing that they can
be adapted to perform in time 
$\bigo(n(1+\entropy(\langle m_1,\ldots,m_r \rangle)))$
where $m_i$ is the number of occurrences of $i$ in the
multiset (note this is totally different from our results, that depend
on the distribution of the lengths of monotone runs).

Each adaptive sorting algorithm in the comparison model yields a
compression scheme for permutations, but the encoding thus defined
does not necessarily support the simple application of the permutation
to a single element without decompressing the whole permutation, nor
the application of the inverse permutation.

\section{Contiguous Monotone Runs}
\label{sec:basic-runs}

Our most fundamental representation takes advantage of permutations
that are formed by a few monotone (ascending or descending) runs.

\begin{definition} \label{def:runs}
  A \emph{down step} of a permutation $\pi$ over $[1..n]$ is a position
  $1 \le i < n$ such that $\pi(i+1)<\pi(i)$.
  An \emph{ascending run} in a permutation $\pi$ is a maximal range of
  consecutive positions $[i..j]$ that does not contain any
  down step.
  Let $d_1,d_2, \ldots,d_k$ be the list of consecutive down steps in $\pi$.
  Then the number of ascending runs of $\pi$ is noted $\nRuns = k+1$,
  and the sequence of the lengths of the ascending runs is noted
  $\vRuns = \langle n_1,n_2, \ldots,n_\nRuns\rangle$, where $n_1=
  d_1,n_2=d_2-d_1, \ldots,n_{\nRuns-1}= d_k-d_{k-1},$ and
  $n_\nRuns=n-d_k$.
  (If $k=0$ then $\nRuns=1$ and $\vRuns=\langle n_1 \rangle = 
   \langle n \rangle$.)
  The notions of \emph{up step} and \emph{descending run} are defined
  similarly.
\end{definition}

For example, the permutation
$(\mathit{1},\mathit{3},\mathit{5},\mathit{7},\mathit{9},\mathbf{2},\mathbf{4},\mathbf{6},\mathbf{8},\mathbf{10})$
contains $\nRuns=2$ ascending runs, of lengths forming the vector $\vRuns=\langle 5,5 \rangle$.

We now describe a data structure that represents a permutation
partitioned into $\nRuns$ ascending runs, and is able to compute 
any $\pi(i)$ and $\pi^{-1}(i)$.

\subsection{Structure} \label{sec:structure}

  \paragraph*{Construction}
  We find the down-steps of $\pi$ in linear time, obtaining $\nRuns$ runs
  of lengths $\vRuns = \langle n_1,\ldots,n_\nRuns\rangle$, and then apply 
  the Huffman algorithm to the vector $\vRuns$.
  When we set up the leaves $v$ of the Huffman tree, 
  we store their original index in $\vRuns$, $idx(v)$, 
  and the starting position in $\pi$ of their corresponding run, $pos(v)$.
  After the tree is built, we use $idx(v)$ to compute a permutation $\phi$ over
  $[1..\nRuns]$ so that $\phi(i)=j$ if the leaf corresponding to
  $n_i$ is placed at the $j$th left-to-right leaf in the Huffman
  tree. We also compute $\phi^{-1}$. 
  We also precompute a bitmap $C[1..n]$ that marks the beginning of runs in 
$\pi$ and give constant-time support for $\BinRank$ and $\BinSelect$. Since 
$C$ contains only $\nRuns$ bits set out of $n$, it is represented in
compressed form \cite{GGGRR07} within $\nRuns \lg\frac{n}{\nRuns} + o(n)$ bits.
   
  Now we set a new permutation $\pi'$ over $[1..n]$
  where the runs are written in the order given by $\phi^{-1}$: 
  We first copy from $\pi$ the run whose endpoints are those
  of the leftmost tree leaf, then the run pointed by the second
  leftmost leaf, and so on. 
  Simultaneously, we compute $pos'(v)$ for the leaves $v$, denoting the 
  starting position of the area they cover in $\pi'$.
  After creating $\pi'$ the original permutation $\pi$ can be
  deleted.
  We say that an internal node {\em covers} the contiguous area of $\pi'$
  formed by concatenating the runs of all the leaves that descend from $v$.
  We compute, for all nodes $v$, 
  $pos'(v)$, the starting position of the area covered by $v$ in $\pi'$,
  $\mathit{length}(v)$, the size of that area, and
  $\mathit{leaves}(v)$, the number of leaves that descend from $v$. 

  Now we enhance the Huffman tree into a wavelet-tree-like
  structure~\cite{GGV03} without altering its shape, as
  follows. 
  Starting from the root, first process recursively each child. 
  For the leaves we do nothing.
  Once the left and right children, $v_l$ and $v_r$, of an internal node $v$ 
  have been processed, the invariant is that the areas they cover have 
  already been sorted.
  We create a bitmap for $v$, of size $\mathit{length}(v)$.
  Now we merge the areas of $v_l$ and $v_r$ in time $\bigo(\mathit{length}(v))$.
  As we do the merging, each time we take an element from $v_l$
  we append a bit $0$ to the node bitmap, and a bit $1$ when we take
  an element from $v_r$.
  When we finish, $\pi'$ has been sorted and we can delete
  it. The Huffman-shaped wavelet tree (only with fields $\mathit{leaves}$ and
  $pos$), $\phi$, and $C$ represent $\pi$.  

  \paragraph*{Space and construction cost}
  Note that each of the $n_i$ elements of leaf $i$ (at depth $\ell_i$) is
  merged $\ell_i$ times, contributing $\ell_i$ bits to the bitmaps of its 
  ancestors, and thus the total number of bits in all bitmaps is 
  $\sum n_i \ell_i$.
  Thus the total number of bits in the Huffman-shaped wavelet tree is at 
  most $n(1+\entropy(\vRuns))$. 
  Those bitmaps, however, are represented in compressed form \cite{GGGRR07},
which allows us removing the $n$ extra bits added by the Huffman encoding.

Let us call $m_j = n_{\phi^{-1}(j)}$ the length of the run corresponding
to the $j$th left-to-right leaf, and $m_{i,j} = m_i + \ldots + m_j$. The 
compressed representation \cite{GGGRR07} takes, on a bitmap of length $n$ and
$m$ 1s, $m\lg\frac{n}{m} + (n-m)\lg\frac{n}{n-m}$ bits, plus a redundancy of 
$\bigo(n\lg\lg n/\lg^2 n)$ bits.
We prove by induction (see also Grossi et al.~\cite{GGV03})
that the compressed space allocated for all the bitmaps 
descending from a node covering leaves $[i..k]$ is 
$\sum_{i \le r \le k} m_r \lg\frac{m_{i,k}}{m_r}$ (we consider the redundancy 
later).
Consider two sibling leaves merging two runs of $m_i$ and $m_{i+1}$ elements. 
Their parent bitmap contains $m_i$ 0s and $m_{i+1}$ 1s, and 
thus its compressed representation requires $m_i\lg\frac{m_i+m_{i+1}}{m_i}
+ m_{i+1}\lg\frac{m_i+m_{i+1}}{m_{i+1}}$ bits.
Now consider a general Huffman tree node merging a left subtree covering 
leaves $[i..j]$ and a right subtree covering leaves $[j+1..k]$. Then the 
bitmap of the node will be compressed to
$m_{i,j}\lg\frac{m_{i,k}}{m_{i,j}} + m_{j+1,k}\lg\frac{m_{i,k}}{m_{j+1,k}}$
bits. By the inductive hypothesis, all the bitmaps on the left child and its
subtrees add up to $\sum_{i \le r \le j} m_r \lg\frac{m_{i,j}}{m_r}$, and
those on the right add up to 
$\sum_{j+1 \le r \le k} m_r \lg\frac{m_{j+1,k}}{m_r}$. Adding up the three
formulas we get the inductive thesis.

Therefore, a compressed representation of the bitmaps requires 
$n\entropy(\vRuns)$ bits, plus the redundancy. The latter, added over all 
the bitmaps, is
$\bigo(n(1+\entropy(\vRuns))\lg\lg n/\lg^2 n) = o(n)$ because
$\entropy(\vRuns) \le \lg n$.%
\footnote{To make sure this is $o(n)$ even if there are many short bitmaps,
we can concatenate all the bitmaps into a single one, and replace pointers
to bitmaps by offsets to this single bitmap. Operations $\BinRank$ and
$\BinSelect$ translate easily into a concatenated bitmap.}
  To this we must add the $\bigo(\nRuns\lg n)$ bits of the tree pointers
  and extra data like $pos$ and $\mathit{leaves}$, the 
  $\bigo(\nRuns\lg \nRuns)$ bits for $\phi$, and the 
  $\nRuns\lg\frac{n}{\nRuns}+o(n)$ bits for $C$.
  
  The construction time is $\bigo(\nRuns\lg\nRuns)$ for the Huffman
  algorithm, plus $\bigo(\nRuns)$ for computing $\phi$ and filling
  the node fields like $pos$ and $\mathit{leaves}$,
  plus $\bigo(n)$ for constructing $\pi'$ and $C$, plus the total 
  number of bits appended to all bitmaps, which includes the merging cost. 
  The extra structures for $\BinRank$ are built in
  linear time on those bitmaps.%
\footnote{While the linear construction time is not obvious from their
article \cite{GGGRR07}, a subsequent result \cite{Pat08} achieved even less
redundancy and linear construction time.}
  All this adds up to $\bigo(n(1+\entropy(\vRuns)))$, because
  $\nRuns \lg\nRuns \le n\entropy(\vRuns) + \lg n$ by
  concavity, recall Definition~\ref{def:entrop}. 

\subsection{Queries}

  \paragraph*{Computing $\pi$ and $\pi^{-1}$}
  One can regard the wavelet tree as a device that tracks the
  evolution of a merge-sorting of $\pi'$, so that in the
  bottom we have (conceptually) the sequence $\pi'$ (with one
  run per leaf) and in the top we have (conceptually) the sorted
  permutation $(1, 2, \ldots, n)$.
  
  To compute $\pi^{-1}(j)$ we start at the top and find out
  where that position came from in $\pi'$. 
  We start at offset $j' = j$ of the root bitmap $B$.  
  If $B[j']=0$, then position $j'$ came from the left subtree in the
  merging. 
  Thus we go down to the left child with $j'
  \leftarrow\BinRank_0(B,j')$, which is the position of $j'$ in the
  array of the left child before the merging.  
  Otherwise we go down to the right child with $j' \leftarrow
  \BinRank_1(B,j')$. We continue recursively until we reach a
  leaf $v$. At this point we know that $j$ came from the corresponding
  run, at offset $j'$, that is, $\pi^{-1}(j) = pos(v)+j'-1$.

  To compute $\pi(i)$ we do the reverse process, but we must
  first determine the leaf $v$ and offset $i'$ within $v$ corresponding
  to position $i$: We compute $l=\phi(\BinRank_1(C,i))$, so that $i$ falls
  at the $l$th left-to-right leaf.
  Then we traverse the Huffman tree down so as to find the
  $l$th leaf. This is easily done as we have $\mathit{leaves}(v)$ stored
  at internal nodes. Upon arriving at leaf $v$, we know that the offset is
  $i' = i - pos(v) + 1$. We now start an
  upward traversal from $v$ using the nodes that are already in the
  recursion stack. If $v$ is a left child of its parent $u$, then we
  set $i' \leftarrow \BinSelect_0(B,i')$ to locate it in the merged
  array of the parent, else we set $i' \leftarrow \BinSelect_1(B,i')$,
  where $B$ is the bitmap of $u$. Then we set $v \leftarrow u$ and
  continue until reaching the root, where we answer $\pi(i) = i'$.

  \paragraph*{Query time}

  In both queries the time is $\bigo(\ell)$, where $\ell$ is the depth
  of the leaf arrived at. If $i$ is chosen uniformly at random in
  $[1..n]$, then the average cost is $\frac{1}{n}\sum n_i\ell_i =
  \bigo(1+\entropy(\vRuns))$. 
  However, the worst case can be $\bigo(\nRuns)$ in a fully
  skewed tree. We can ensure $\ell = \bigo(\lg\nRuns)$ in the
  worst case while maintaining the average case by slightly
  rebalancing the Huffman tree \cite{ML01}.
  Given any constant $x>0$, the height of the Huffman tree can be bound
  to at most $(1+x)\lg\nRuns$ so that the total number of
  bits added to the encoding is at most $n \cdot
  \nRuns^{-x\lg\varphi}$, where $\varphi \approx 1.618$ is the
  golden ratio. This is $o(n)$ if $\nRuns=\omega(1)$, and
  otherwise the cost was $\bigo(\nRuns)=\bigo(1)$ anyway. 
  Similarly, the average time stays
  $\bigo(1+\entropy(\vRuns))$, as it increases at most by
  $\bigo(\nRuns^{-x\lg\varphi}) = \bigo(1)$.
  This rebalancing takes just $\bigo(\nRuns)$ time if the frequencies are 
  already sorted.

  Note also that the space required by the query is
  $\bigo(\lg\nRuns)$.  
  This can be made constant by storing parent pointers in the wavelet
  tree, which does not change the asymptotic space.

\begin{theorem}
  There is an encoding scheme using at most
  $n\entropy(\vRuns) + \bigo(\nRuns\lg n) + o(n)$ bits to represent
  a permutation $\pi$ over $[1..n]$ covered by $\nRuns$ contiguous 
  ascending runs of lengths forming the vector $\vRuns$.
  It can be built within time $\bigo(n(1+\entropy(\vRuns)))$, and
  supports the computation of $\pi(i)$ and $\pi^{-1}(i)$ in time
  $\bigo(1+\lg\nRuns)$ for any value of $i\in[1..n]$.
  If $i$ is chosen uniformly at random in $[1..n]$ then the average
  computation time is $\bigo(1+\entropy(\vRuns))$.
\label{thm:main}
\end{theorem}

We note that the space analysis leading to $n\entropy(\vRuns)+o(n)$ bits
works for any tree shape. We could have used a balanced tree, yet we would
not achieve $\bigo(1+\entropy(\vRuns))$ average time. On the other hand,
by using Hu-Tucker codes instead of Huffman, as in our previous work 
\cite{BN09}, we would not need the permutation $\phi$ and, by using
compact tree representations \cite{SN10}, we would be able to reduce the
space to $n\entropy(\vRuns) + \bigo(\nRuns\lg\frac{n}{\nRuns}) +
o(n)$. This is interesting for large values of $\nRuns$, as it
is always $n\entropy(\vRuns) + o(n(1+\entropy(\vRuns))$ even if
$\nRuns = \Theta(n)$.%
\footnote{We do not follow this path because we are more interested in
multiary codes (see Section~\ref{sec:multiary}) and, to the best of our
knowledge, there is no efficient (i.e., $\bigo(\nRuns\lg\nRuns)$ time)
algorithm for building multiary Hu-Tucker codes \cite{theArtOfComputerProgrammingVol3}.}

\subsection{Mixing Ascending and Descending Runs}

We can easily extend Theorem~\ref{thm:main} to mix ascending and descending
runs.

\begin{corollary} \label{cor:main}
Theorem~\ref{thm:main} holds verbatim if $\pi$ is partitioned into a sequence
$\nRuns$ contiguous monotone 
(i.e., ascending or descending) runs of lengths forming
the vector $\vRuns$.
\end{corollary}
\begin{proof}
We mark in a bitmap of length $\nRuns$ whether each run is ascending or 
descending, and then reverse descending runs in $\pi$, so as to obtain
a new permutation $\pi_{asc}$, which is represented using 
Theorem~\ref{thm:main} (some runs of $\pi$ could now be merged in 
$\pi_{asc}$, but this only reduces $\entropy(\vRuns)$, recall 
Definition~\ref{def:entrop}). 

The values $\pi(i)$ and $\pi^{-1}(j)$ are easily computed from
$\pi_{asc}$:
If $\pi^{-1}_{asc}(j) = i$, we use $C$ to determine that $i$ is within
run $\pi_{asc}(\ell..r)$, that is, $\ell = \BinSelect_1(\BinRank_1(C,i))$
and $r = \BinSelect_1(\BinRank_1(C,i)+1)-1$.
If that run is reversed in $\pi$, then $\pi^{-1}(j) = \ell+r-i$, else
$\pi^{-1}(j)=i$.
For $\pi(i)$, we use $C$ to determine that $i$ belongs to run
$\pi(\ell..r)$.  If the run is descending, then we return
$\pi_{asc}(\ell+r-i)$, else we return $\pi_{asc}(i)$.
The operations on $C$ require only constant time.  
The extra construction time is just $\bigo(n)$,
and no extra space is needed apart from $\nRuns = o(\nRuns\lg n)$ bits. 
\end{proof}

Note that, unlike the case of ascending runs, where there is an obviously 
optimal way of partitioning (that is, maximize the run lengths), we have some
freedom when partitioning into ascending or descending runs, at the endpoints
of the runs: If an ascending (resp. descending) run is followed by a 
descending (resp. ascending) run, the limiting element can be moved to either run;
if two ascending (resp. descending) runs are consecutive, one can create 
a new descending (resp. ascending) run with the two endpoint elements. While finding
the optimal partitioning might not be easy, we note that these decisions cannot
affect more than $\bigo(\nRuns)$ elements, and thus the entropy of the partition
cannot be modified by more than $\bigo(\nRuns\lg n)$, which is absorbed by the
redundancy of our representation.

\subsection{Improved Adaptive Sorting}

One of the best known sorting algorithms is MergeSort, based
on a simple linear procedure to merge two already sorted arrays,
and with a worst case complexity of $n\lceil\lg n\rceil$
comparisons and $\bigo(n\lg n)$ running time.
It had been already noted~\cite{theArtOfComputerProgrammingVol3}
that finding the down-steps of the array in linear time allows improving
the time of MergeSort to $\bigo(n(1+\lg\nRuns))$ (the down-step concept can 
be applied to general sequences, where consecutive equal values do not break 
runs).

We now show that the construction process of our data structure sorts the 
permutation and, applied on a general sequence, it achieves a refined 
sorting time of
$\bigo(n(1+\entropy(\vRuns))\subset\bigo(n(1+\lg\nRuns))$ (since
$\entropy(\vRuns)\le\lg\nRuns$).

\begin{theorem}
  There is an algorithm sorting an array of length $n$ covered by
  $\nRuns$ contiguous monotone runs of lengths forming the vector  $\vRuns$ in time
  $\bigo(n(1+\entropy(\vRuns)))$, which is worst-case optimal in the
  comparison model.
  \label{thm:mainsort}
\end{theorem}
\begin{proof}
  Our wavelet tree construction of Theorem~\ref{thm:main} (and
  Corollary~\ref{cor:main}) indeed sorts
  $\pi$ within this time, and it also works if the array is not a
  permutation.
  This is optimal because, even considering just ascending
  runs, there are $\frac{n!}{n_1!n_2!\ldots n_\nRuns!}$ different
  permutations that can be covered with runs of lengths forming the
  vector $\vRuns = \langle n_1,n_2,\ldots, n_\nRuns\rangle$.
  Thus $\lg \frac{n!}{n_1!n_2! \ldots n_\nRuns!}$ comparisons
  are necessary. Using Stirling's approximation to the factorial we have
  $\lg \frac{n!}{n_1!n_2! \ldots n_\nRuns!}=
  (n+1/2)\lg n - \sum_i (n_i+1/2)\lg n_i - \bigo(\lg \nRuns)$. Since
  $\sum \lg n_i \le \nRuns \lg(n/\nRuns)$, this is
  $n\entropy(\vRuns)-\bigo(\nRuns\lg(n/\nRuns)) = n\entropy(\vRuns)-\bigo(n)$.
  The term $\Omega(n)$ is also necessary to read the input, hence
  implying a lower bound of $\Omega(n(1+\entropy(\vRuns)))$.

Note, however, that the set of permutations that {\em can be} covered with 
$\nRuns$ runs of lengths $\vRuns$, may contain permutations that can be
covered with fewer runs (as two consecutive runs could be merged), and thus
they have entropy less than $\entropy(\vRuns)$, recall
Definition~\ref{def:entrop}. We have proved that the lower
bound applies to the union of two classes: one (1) contains (some%
\footnote{Other permutations with vectors distinct from $\vRuns$ could also
have entropy $\entropy(\vRuns)$.}) permutations
of entropy $\entropy(\vRuns)$ and the other (2) contains (some) permutations 
of entropy less than $\entropy(\vRuns)$. Obviously the bound does not hold for 
class (2) alone, as we can sort it in less time. 
Since we can tell the class of a permutation in $\bigo(n)$ 
time by counting the down-steps, it follows that the bound also applies to 
class (1) alone (otherwise $\bigo(n) + o(n\entropy(\vRuns))$ would be 
achievable for (1)$+$(2)).
\end{proof}

\subsection{Boosting Time Performance} \label{sec:multiary}

The time performance achieved in Theorem~\ref{thm:main} (and
Corollary~\ref{cor:main}) can be boosted
by an $\bigo(\lg\lg n)$ time factor by using Huffman codes of higher
arity. 

Given the run lengths $\vRuns$, we build the $t$-ary Huffman tree for
$\vRuns$, with $t=\sqrt{\lg n}$. Since now we merge $t$ children to build the 
parent, the sequence stored in the parent to indicate the child each element 
comes from is not binary, but over alphabet $[1..t]$. In addition, we set up
$\nRuns$ pointers to provide direct access to the leaves, and parent 
pointers.

The total length of all the sequences stored at all the Huffman tree nodes
is $< n(1+\entropy(\vRuns)/\lg t)$ \cite{Huf52}. 
To reduce the redundancy, we represent
each sequence $S[1..m]$ stored at a node using the compressed representation
of Golynski et al.~\cite[Lemma 9]{GRR08}, which yields space 
$m\entropy_0(S)+\bigo(m\lg t \lg\lg m/ \lg^2 m)$ bits.

For the string $S[1..m]$ corresponding to a leaf covering run lengths
$m_1, \ldots, m_t$, we have $m\entropy_0(S) = \sum m_i \lg\frac{m}{m_i}$.
From there we can carry out exactly the same analysis done in
Section~\ref{sec:structure} for binary trees, to conclude that the sum
of the $m\entropy_0(S)$ bits for all the strings $S$ over all the tree
nodes is $n\entropy(\vRuns)$. On the other hand, the redundancies add up to
$\bigo(n(1+\entropy(\vRuns)/\lg t)\lg t \lg\lg n / \lg^2 n) = o(n)$ bits.%
\footnote{Again, we can concatenate all the sequences to make sure this
redundancy is asymptotic in $n$.}

The advantage of the $t$-ary representation is that the average leaf depth
is $1+\entropy(\vRuns)/\lg t = \bigo(1+\entropy(\vRuns)/\lg\lg n)$. The
algorithms to compute $\pi(i)$ and $\pi^{-1}(i)$ are similar, except that
$\BinRank$ and $\BinSelect$ are carried out on sequences $S$ over alphabets of
size $\sqrt{\lg n}$. Those operations can still be carried out in constant
time on the representation we have chosen \cite{GRR08}. The only detail is
that, for $\pi(i)$ we first moved from the root to the leaf using the field
$\mathit{leaves}(v)$. This does not anymore allow us processing a node in
constant time, and thus we have opted for storing an array of pointers to 
the leaves and parent pointers.

For the worst case, if $\nRuns=\omega(1)$, we can again limit the depth of the 
Huffman tree to $\bigo(\lg \nRuns / \lg\lg n)$ and maintain the same average 
time. The multiary case is far less understood than the binary case. 
Recently, an algorithm to find the optimal length-restricted $t$-ary code 
has been presented whose running time is linear once the lengths are sorted
\cite{Bae07}. To analyze the increase in redundancy, consider the sub-optimal
method that simply takes any node $v$ of depth more than $\ell=4\lg\nRuns /
\lg t$ and balances its subtree (so that height $5 \lg \nRuns / \lg t$ is 
guaranteed). Since any node at depth $\ell$ covers a total length of at most
$n/t^{\lfloor \ell/2\rfloor}$ (see next paragraph), the sum of all the lengths 
covered by these nodes is at most $\nRuns\cdot n/t^{\lfloor \ell/2\rfloor}$. 
By forcing those subtrees to be balanced, the average leaf depth increases by 
at most $(\lg\nRuns/\lg t)~\nRuns/t^{\lfloor \ell/2\rfloor} \le
 \lg(\nRuns)/ (\nRuns \lg t) = \bigo(1)$. Hence the worst case is limited to 
$\bigo(1+\lg \nRuns /\lg\lg n)$ while the average case stays
within $\bigo(1+\entropy(\vRuns)/\lg\lg n)$. For the space we need a finer
consideration: As $\nRuns=\omega(1)$, the increase in average leaf depth is 
$o(1/\lg t)$. Since increasing by one the depth of a leaf covering $m$ 
elements costs $m\lg t$ further bits, the total increase in space redundancy 
is $o(n)$. 

The limit on the probability is obtained as follows. Consider a node $v$ in
the $t$-ary Huffman tree. Then $\mathit{length}(u) \ge \mathit{length}(v)$ for 
any uncle $u$ of $v$, as otherwise switching $v$ and $u$ improves the already 
optimal Huffman tree. Hence $w$, the grandparent of $v$ (i.e., the parent of 
$u$) must cover an area of size $\mathit{length}(w) \ge t \cdot
\mathit{length}(v)$. Thus the covered length is 
multiplied at least by $t$ when moving from a node to its grandparent. 
Conversely, it is divided at least by $t$ as we move from a node to any 
grandchild. As the total length at the root is $n$, the length covered by
any node $v$ at depth $\ell$ is at most $\mathit{length}(v) \le
n/t^{\lfloor \ell/2\rfloor}$.

This yields our final result for contiguous monotone runs.

\begin{theorem}
  There is an encoding scheme using at most
  $n\entropy(\vRuns) + \bigo(\nRuns\lg n) + o(n)$ bits to encode
  a permutation $\pi$ over $[1..n]$ covered by $\nRuns$ contiguous monotone
  runs of lengths forming the vector $\vRuns$.
  It can be built within time $\bigo(n(1+\entropy(\vRuns)/\lg\lg n))$, and
  supports the computation of $\pi(i)$ and $\pi^{-1}(i)$ in time
  $\bigo(1+\lg\nRuns / \lg\lg n)$ for any value of $i\in[1..n]$.
  If $i$ is chosen uniformly at random in $[1..n]$ then the average
  computation time is $\bigo(1+\entropy(\vRuns)/\lg\lg n)$.
\label{thm:main2}
\end{theorem}

The only missing part is the construction time, since now we have to build
strings $S[1..m]$ by merging $t$ increasing runs. This can be done in 
$\bigo(m)$ time by using atomic heaps~\cite{FW94atomic}. The compressed
sequence representations are built in linear time \cite{GRR08}. 
Note this implies
that we can sort an array with $\nRuns$ contiguous monotone runs of lengths
forming the vector $\vRuns$ in time $\bigo(n(1+\entropy(\vRuns)/\lg\lg n))$,
yet we are not anymore within the comparison model.

\subsection{An Improved Sequence Representation}

Interestingly, the previous result yields almost directly a new representation 
of 
sequences that, compared to the state of the art \cite{FMMN07,GRR08}, provides 
improved average time performance.

\begin{theorem}
  Given a string $S[1..n]$ over alphabet $[1..\sigma]$ with zero-order entropy
  $\entropy_0(S)$, there is an encoding for $S$ using at most
  $n\entropy_0(S) + \bigo(\sigma\lg n) + o(n)$ bits and answering
  queries $S[i]$, $\StrRank_c(S,i)$ and $\StrSelect_c(S,i)$ in time
  $\bigo(1+\lg\sigma/\lg\lg n)$ for any $c \in [1..\sigma]$ and $i \in [1..n]$.
  When $i$ is chosen at random in query $S[i]$, or $c$ is chosen with
  probability $n_c/n$ in queries $\StrRank_c(S,i)$ and $\StrSelect_c(S,i)$, 
  where $n_c$ is the frequency of $c$ in $S$, the average query time is
  $\bigo(1+\entropy_0(S)/\lg\lg n)$.
\label{thm:string}
\end{theorem}
\begin{proof}
We build exactly the same $t$-ary Huffman tree used in Theorem~\ref{thm:main2},
using the frequencies $n_c$ instead of run lengths. The sequences at each
internal node are formed so as to indicate how the symbols in the child nodes
are interleaved in $S$. This is precisely a multiary Huffman-shaped wavelet
tree \cite{GGV03,FMMN07}, and our previous analysis shows that the space used
by the tree is exactly as in Theorem~\ref{thm:main2}, where now the entropy is 
$\entropy_0(S) = \sum_c \frac{n_c}{n} \lg\frac{n}{n_c}$. The three queries are 
solved by going down or up the tree and using $\StrRank$ and $\StrSelect$ on 
the sequences stored at the nodes \cite{GGV03,FMMN07}. Under the conditions
stated for the average case, one arrives at the leaf of symbol $c$ with
probability $n_c/n$, and then the average case complexities follow.
\end{proof}

\section{Strict Runs}
\label{sec:stricter}

Some classes of permutations can be covered by a small number of
runs of a stricter type.
We present an encoding scheme that take advantage of them.

\begin{definition}
  A \emph{strict ascending run} in a permutation $\pi$ is a maximal
  range of positions satisfying $\pi(i+k)=\pi(i)+k$. The {\em head} of
  such run is its first position.
  The number of strict ascending runs of $\pi$ is noted $\nSRuns$,
  and the sequence of the lengths of the strict ascending runs is
  noted $\vSRuns$. 
  We will call $\vHRuns$ the sequence of contiguous monotone
  run lengths of the sequence
  formed by the strict run heads of $\pi$.
Similarly, the notion of a \emph{strict descending run} can be defined,
as well as that of \emph{strict (monotone) run} encompassing both.
\end{definition}

For example, the permutation
$(\mathit{6},\mathit{7},\mathit{8},\mathit{9},\mathit{10},\mathbf{1},\mathbf{2},\mathbf{3},\mathbf{4},\mathbf{5})$
contains $\nSRuns=2$ strict runs, of lengths $\vSRuns = \langle 5,5
\rangle$.
The run heads are $\langle \mathit{6}, \mathbf{1} \rangle$,
which form 1 monotone run, of lengths $\vHRuns = \langle 2
\rangle$. 
Instead, the permutation
$(\mathit{1},\mathit{3},\mathit{5},\mathit{7},\mathit{9},\mathbf{2},\mathbf{4},\mathbf{6},\mathbf{8},\mathbf{10})$
contains $\nSRuns=10$ strict runs, each of length 1.

\begin{theorem}
\label{thm:strict}
Assume there is an encoding $P$ for a permutation over $[1..n]$ with
$\nRuns$ contiguous monotone runs of lengths forming the vector  $\vRuns$, which requires
$s(n,\nRuns,\vRuns)$ bits of space and can apply the permutation and
its inverse in time $t(n,\nRuns,\vRuns)$.
Now consider a permutation $\pi$ over $[1..n]$ covered by $\nSRuns$
strict runs and by $\nRuns\le\nSRuns$ monotone runs, and let $\vHRuns$ be the
vector formed by the $\nRuns$ monotone run lengths in the permutation of strict
run heads.
Then there is an encoding scheme using at most 
  $s(\nSRuns,\nRuns,\vHRuns) + \bigo(\nSRuns\lg\frac{n}{\nSRuns})
  + o(n)$ bits  for $\pi$.
  It can be computed in $\bigo(n)$ time on top of that for building $P$.
  It supports the computation of $\pi(i)$ and $\pi^{-1}(i)$ in time
  $\bigo(t(\nSRuns,\nRuns,\vHRuns))$ for any value $i\in[1..n]$.
\end{theorem}
\begin{proof}
We first set up a bitmap $R$ of length $n$ marking with a 1 bit the beginning 
of the strict runs. 
We set up a second bitmap $R^{inv}$ such that 
$R^{inv}[i] = R[\pi^{-1}(i)]$.
Now we create a new permutation $\pi'$ over $[1..\nSRuns]$ which collapses the 
strict runs of $\pi$, $\pi'(i) = \BinRank_1(R^{inv},\pi(\BinSelect_1(R,i)))$.
All this takes $\bigo(n)$ time and the bitmaps take 
$2\nSRuns\lg\frac{n}{\nSRuns} + \bigo(\nSRuns) + o(n)$ bits in compressed
form \cite{GGGRR07}, where $\BinRank$ and $\BinSelect$ are supported in 
constant time.

Now we build the structure $P$ for $\pi'$. The number of monotone runs in
$\pi$ is the same as for the sequence of strict run heads in $\pi$,
and in turn the same as the runs in $\pi'$. So the number of
runs in $\pi'$ is also $\nRuns$ and their lengths are $\vHRuns$.
Thus we require $s(\nSRuns,\nRuns,\vHRuns)$ further bits.

To compute $\pi(i)$, we find $i' \leftarrow \BinRank_1(R,i)$ and then compute 
$j' \leftarrow \pi'(i')$. The final answer is 
$\BinSelect_1(R^{inv},j') + i-\BinSelect_1(R,i')$.
To compute $\pi^{-1}(j)$, we find $j' \leftarrow \BinRank_1(R^{inv},j)$ and then
compute $i' \leftarrow (\pi')^{-1}(j')$. The final answer is 
$\BinSelect_1(R,i') + j-\BinSelect_1(R^{inv},j')$. 
The structure requires only constant time on top of that to support
the operator $\pi'()$ and its inverse $\pi'^{-1}()$ .
\end{proof}

The theorem can be combined with previous results, for example
Theorem~\ref{thm:main2}, in order to obtain
concrete data structures.  This representation is interesting because
its space could be much less than $n$ if $\nSRuns$ is small enough. However,
it still retains an $o(n)$ term that can be dominant.
The following corollary describes a compressed data structure where
the $o(n)$ term is significantly reduced.

\begin{corollary}
  The $o(n)$ term in the space of Theorem~\ref{thm:strict} can be replaced by 
  $\bigo(\nSRuns\llg\frac{n}{\nSRuns} + \lg n)$ at the cost of
  $\bigo(1+\lg \nSRuns)$ extra time for the queries.
\label{cor:bsgap}
\end{corollary}
\begin{proof}
Replace the structure of Golynski et al.~\cite{GGGRR07} by the binary searchable gap 
encoding of Gupta et al.~\cite{GHSV06}, which takes $\bigo(1+\lg \nSRuns)$ time 
for $\BinRank$ and $\BinSelect$ (recall Section~\ref{sec:sequences}).
\end{proof}

Other tradeoffs for the bitmap encodings are possible, such as the one
described by Gupta~\cite[Theorem~18 p.~155]{Gup07}.

\section{Shuffled Sequences}
\label{sec:shuffled-upsequences}

Up to now our runs have been contiguous in $\pi$.
Levcopoulos and Petersson~\cite{sortingShuffledMonotoneSequences}
introduced the more sophisticated concept of partitions formed by
interleaved runs, such as \emph{Shuffled UpSequences} (SUS)
and \emph{Shuffled Monotone Sequences} (SMS).
We now show how to take advantage of permutations
formed by shuffling (interleaving) a small number of runs.

\begin{definition}
  A decomposition of a permutation $\pi$ over $[1..n]$ into
  \emph{Shuffled UpSequences} is a set of, not necessarily consecutive,
  subsequences of increasing numbers that have to be removed from $\pi$ 
  in order to reduce it to the empty sequence.
  The number of shuffled upsequences in such a decomposition of $\pi$
  is noted $\nSUS$, and the vector formed by the lengths of the
  involved shuffled upsequences, in arbitrary order, is noted $\vSUS$.
  When the subsequences can be of increasing or decreasing numbers, we
  call them \emph{Shuffled Monotone Sequences}, call $\nSMS$ their
  number and $\vSMS$ the vector formed by their lengths.
\end{definition}

For example, the permutation
$(\mathit{1},\mathbf{6},\mathit{2},\mathbf{7},\mathit{3},\mathbf{8},\mathit{4},\mathbf{9},\mathit{5},\mathbf{10})$
contains $\nSUS=2$ shuffled upsequences of lengths forming the vector
$\vSUS=\langle 5,5\rangle$, but $\nRuns=5$ runs, all of
length 2.
Interestingly, we can reduce the problem of representing shuffled sequences to 
that of representing strings and contiguous runs.

\subsection{Reduction to Strings and Contiguous Monotone Sequences}

We first show how a permutation with a small number of shuffled monotone
sequences can be represented using strings over a small alphabet and 
permutations with a small number of contiguous monotone sequences.

\begin{theorem} \label{thm:gralSMS} 
  Assume there exists an encoding $P$ for a permutation over $[1..n]$ with 
  $\nRuns$ contiguous monotone runs of lengths forming the vector $\vRuns$, which 
  requires $s(n,\nRuns,\vRuns)$ bits of space and
  supports the application of the permutation and its inverse in time
  $t(n,\nRuns,\vRuns)$.
  Assume also that there is a data structure $S$ for a string $S[1..n]$
  over an alphabet of size $\nSMS$ with symbol frequencies $\vSMS$,
  using $s'(n,\nSMS,\vSMS)$ bits of space and
  supporting operators $\StrRank$, $\StrSelect$, and access to values
  $S[i]$, in time $t'(n,\nSMS,\vSMS)$.
  Now consider a permutation $\pi$ over $[1..n]$ covered by $\nSMS$
  shuffled monotone sequences of lengths $\vSMS$.
  Then there exists an encoding of $\pi$ using at most
  $s(n,\nSMS,\vSMS)+s'(n,\nSMS,\vSMS) +
  \bigo(\nSMS\lg\frac{n}{\nSMS})+o(n)$ bits.
  Given the covering into SMSs, the encoding can be built in time
  $\bigo(n)$, in addition to that of building $P$ and $S$.
  It supports the computation of $\pi(i)$ and $\pi^{-1}(i)$ in time
  $t(n,\nSMS,\vSMS)+t'(n,\nSMS,\vSMS)$ for any value of $i\in[1..n]$.
  The result is also valid for shuffled upsequences, in which case $P$
  is just required to handle ascending runs.
\end{theorem}
\begin{proof}
  Given the partition of $\pi$ into $\nSMS$ monotone subsequences, we
  create a string $S[1..n]$ over alphabet $[1..\nSMS]$ that indicates,
  for each element of $\pi$, the label of the monotone sequence it
  belongs to.
  We encode $S[1..n]$ using the data structure $S$.
  We also store an array $A[1..\nSMS]$ so that $A[\ell]$ is the
  accumulated length of all the sequences with label less than $\ell$.
  
  Now consider the permutation $\pi'$ formed by the sequences taken in
  label order: $\pi'$ can be covered with $\nSMS$ contiguous
  monotone runs $\vSMS$, and hence can be encoded using
  $s(n,\nSMS,\vSMS)$ additional bits using $P$.
  This supports the operators $\pi'()$ and $\pi'^{-1}()$ in time
  $t(n,\nSMS,\vSMS)$
 (again, some of the runs
  could be merged in $\pi'$, which only improves time and space in $P$).  
  Thus $\pi(i) = \pi'(A[S[i]]+\StrRank_{S[i]}(S,i))$ can be computed
  in time $t(n,\nSMS,\vSMS)+t'(n,\nSMS,\vSMS)$. Similarly,
  $\pi^{-1}(i) = \StrSelect_\ell(S,(\pi')^{-1}(i)-A[\ell])$, where
  $\ell$ is such that $A[\ell] < (\pi')^{-1}(i) \le A[\ell+1]$, can
  also be computed in time $t(n,\nSMS,\vSMS)+t'(n,\nSMS,\vSMS)$, plus
  the time to find $\ell$.
  The latter is reduced to constant by representing $A$ with a bitmap
  $A'[1..n]$ with the bits set at the values $A[\ell]+1$, so that
  $A[\ell] = \BinSelect_1(A',\ell)-1$, and the binary search is
  replaced by $\ell = \BinRank_1(A',(\pi')^{-1}(i))$. With the structure
  of Golynski et al.~\cite{GGGRR07}, $A'$ uses
  $\bigo(\nSMS\lg \frac{n}{\nSMS})+o(n)$ bits and operates in constant
  time.
\end{proof}

We will now obtain concrete results by using specific representations for
$P$ and $S$, and specific methods to find the decomposition into shuffled
sequences.

\subsection{Shuffled UpSequences}

Given an arbitrary permutation, one can decompose it in linear time
into contiguous runs in order to minimize $\entropy(\vRuns)$, where
$\vRuns$ is the vector of run lengths.
However, decomposing the same permutation into shuffled up 
(resp. monotone) sequences so as to minimize either $\nSUS$ or 
$\entropy(\vSUS)$ (resp. $\nSMS$ or $\entropy(\vSMS)$) is
computationally harder.

Fredman~\cite{onComputingTheLengthOfLongestIncreasingSubsequences}
gave an algorithm to compute a partition of minimum size $\nSUS$, into
upsequences, claiming a worst case complexity of $\bigo(n\lg n)$.
Even though he did not claim it at the time, it is easy to observe
that his algorithm is adaptive in $\nSUS$ and takes
$\bigo(n(1+\lg\nSUS))$ time.
We give here an improvement of his algorithm that computes the
partition itself within time $\bigo(n(1+\entropy(\vSUS)))$, no worse
than the time of his original algorithm, as
$\entropy(\vSUS) \le \lg \nSUS$.

\begin{theorem}
  \label{thm:partitionInSUS}
  If an array $D[1..n]$ can be optimally covered by $\nSUS$ shuffled
  upsequences (equal values do not break an upsequence), then there is
  an algorithm finding a covering of size $\nSUS$ in time
  $\bigo(n(1+\entropy(\vSUS)))\subset\bigo(n(1+\lg\nSUS))$, where
  $\vSUS$ is the vector formed by the lengths of the upsequences
  found.
\end{theorem}
\begin{proof}
  Initialize a sequence $S_1=(D[1])$, and a splay tree $T$ \cite{ST85}
  with the node $(S_1)$, ordered by the rightmost value of
  the sequence contained by each node.
  For each further array element $D[i]$, search for the sequence with the
  maximum ending point no larger than $D[i]$.
  If it exists, add $D[i]$ to this sequence, otherwise create a new
  sequence and add it to $T$.

  Fredman~\cite{onComputingTheLengthOfLongestIncreasingSubsequences}
  already proved that this algorithm finds a partition of minimum size $\nSUS$.
  Note that, although the rightmost values of the splay tree nodes
  change when we insert a new element in their sequence, their
  relative position with respect to the other nodes remains the same,
  since all the nodes at the right hold larger values than the one
  inserted.
  This implies in particular that only searches and insertions are
  performed in the splay tree.
  
  A simple analysis, valid for both the plain sorted array in
  Fredman's proof and the splay tree of our own proof, yields an
  adaptive complexity of $\bigo(n(1+\lg\nSUS))$ comparisons, since
  both structures contain at most $\nSUS$ elements at any time.
  The additional linear term (relevant when $\nSUS=1$)
   corresponds to the cost of reading each element once.

  The analysis of the algorithm using the splay tree refines the
  complexity to $\bigo(n(1+\entropy(\vSUS)))$, where
  $\vSUS$ is the vector formed by the lengths of the upsequences
  found. 
  These lengths correspond to the frequencies of access to each
  node of the splay tree, which yields the total access time of
  $\bigo(n(1+\entropy(\vSUS)))$~\cite[Theorem~2]{ST85}.
\end{proof}

The theorem obviously applies to the particular case where the array is a
permutation. For permutations and, in general, integer arrays over a universe
$[1..m]$, we can deviate from the comparison model and find the partition
within time $\bigo(n\llg m)$, by using $y$-fast tries \cite{Wil83} instead
of splay trees.

We can now give a concrete representation for shuffled upsequences.
The complete description of the permutation requires to encode the
computation the partitioning and of the comparisons performed by the
sorting algorithm. This time the encoding cost of partitioning is
as important as that of merging.

\begin{theorem} \label{thm:SUS} Let $\pi$ be a permutation over
  $[1..n]$ that can be optimally covered by $\nSUS$ shuffled
  upsequences, and let $\vSUS$ be the vector formed by the lengths of the
  decomposition found by the algorithm of
  Theorem~\ref{thm:partitionInSUS}.
Then there is an encoding scheme for $\pi$ using at most 
$2n\entropy(\vSUS)+ \bigo(\nSUS\lg n) + o(n)$ bits.
It can be computed in time 
$\bigo(n(1+\entropy(\vSUS)))$, and
supports the computation of $\pi(i)$ and $\pi^{-1}(i)$ in time
$\bigo(1+\lg\nSUS/\lg\lg n)$ for any value of $i\in[1..n]$.
  If $i$ is chosen uniformly at random in $[1..n]$ the average query time
  is $\bigo(1+\entropy(\vSUS)/\lg\lg n)$.
\end{theorem}
\begin{proof}
  We first use Theorem~\ref{thm:partitionInSUS} to find the SUS
  partition of optimal size $\nSUS$, and the corresponding vector
  $\vSUS$ formed by the sizes of the subsequences of this partition.
  Then we apply Theorem~\ref{thm:gralSMS}: For the data structure $S$
  we use Theorem~\ref{thm:string}, whereas for $P$ we use
  Theorem~\ref{thm:main2}. Note $\entropy(\vSUS)$ is both $\entropy_0(S)$
  and $\entropy(\vRuns)$ for permutation $\pi'$. 
  The result follows immediately.
\end{proof}

One would be tempted to consider the case of a permutation $\pi$ covered by 
$\nSUS$ upsequences which form strict runs, as a particular case. Yet, this
is achieved by resorting directly to Theorem~\ref{thm:main2}.
The corollary extends verbatim to shuffled monotone sequences.

\begin{corollary}
  There is an encoding scheme 
  using at most $n\entropy(\vSUS) + \bigo(\nSUS\lg n) +o(n)$
  bits to 
  encode a permutation $\pi$ over $[1..n]$ optimally covered by
  $\nSUS$ shuffled upsequences,
  of lengths forming the vector $\vSUS$,
  and made up of strict runs.
  It can be built within time $\bigo(n(1+\entropy(\vSUS)/\lg\lg n))$, and
  supports the computation of $\pi(i)$ and $\pi^{-1}(i)$ in time
  $\bigo(1+\lg\nSUS /\lg\lg n)$ for any value of $i\in[1..n]$.
  If $i$ is chosen uniformly at random in $[1..n]$ then the average query time
  is $\bigo(1+\entropy(\vSUS)/\lg\lg n)$.
\label{cor:SUSstrict}
\end{corollary}
\begin{proof}
  It is sufficient to invert $\pi$ and represent $\pi^{-1}$ using
  Theorem~\ref{thm:main2}, since in this case $\pi^{-1}$ is covered by
  $\nSUS$ ascending runs of lengths forming the vector $\vSUS$: If
  $i_0 < i_1 \ldots < i_m$ forms a strict upsequence, so that
  $\pi(i_t) = \pi(i_0)+t$, then calling $j_0 = \pi(i_0)$ we have the
  ascending run $\pi^{-1}(j_0+t) = i_t$ for~$0\le t\le m$.
\end{proof}

Once more, our construction translates into an improved sorting
algorithm, improving on the complexity $\bigo(n(1+\lg\nSUS))$ of
the algorithm by Levcopoulos and Petersson~\cite{sortingShuffledMonotoneSequences}.

\begin{corollary}
  \label{cor:sortSUS} 
  We can sort an array of length $n$, optimally covered by $\nSUS$
  shuffled upsequences, in time $\bigo(n(1+\entropy(\vSUS)))$, where
  $\vSUS$ are the lengths of the decomposition found by the algorithm
  of Theorem~\ref{thm:partitionInSUS}.
\end{corollary}
\begin{proof}
  Our construction in Theorem~\ref{thm:SUS} finds and separates the
  subsequences of $\pi$, and sorts them, all within this
  time (we do not need to build the string $S$).
\end{proof}

\paragraph*{Open problem}
Note that the algorithm of Theorem~\ref{thm:partitionInSUS} finds
a partition of minimal size $\nSUS$ (this is what we refer to with
``optimally covered''), but that the entropy
$\entropy(\vSUS)$ of this partition is not necessarily minimal:
There could be another partition, even of size larger than $\nSUS$, with
lower entropy. Our results are only in function of the entropy of the
partition of minimal size $\nSUS$ found. This is unsatisfactory, as the 
ideal would be to speak in terms of the minimum possible
$\entropy(\vSUS)$, just as we could do for $\entropy(\vRuns)$.

An example, consider the
permutation $(1,2,\dots,n/2{-}1,n,n/2,n/2{+}1,\ldots,n{-}1)$, for some even
integer~$n$.
The algorithm of Theorem~\ref{thm:partitionInSUS} yields the partition
$\{(1,2,\dots,n/2{-}1,n),(n/2,n/2{+}1,\ldots,n{-}1)\}$ of entropy
$\entropy(\langle n/2,n/2\rangle)=n\lg 2=n$.
This is suboptimal, as the partition 
$\{(1,2,\dots,n/2{-}1,n/2,n/2{+}1, \ldots,$ $n{-}1),(n)\}$ is of much
smaller entropy, $\entropy(\langle
n{-}1,1\rangle)=(n-1)\lg\frac{n}{n-1} + \lg n =\bigo(\lg n)$.

On the other hand, a greedy online algorithm cannot minimize the
entropy of a SUS partitioning. As an example consider the
permutation $(2,3,\ldots,n/2,1,n,n/2{+}1,\dots,n{-}1)$, for some
even integer~$n$.
A greedy online algorithm that after processing a prefix of the
sequence minimizes the entropy of such prefix, produces the partition
$\{(1,n/2{+}1,\dots,n{-}1),(2,3,\ldots,n/2,n)\}$, of size $2$ and
entropy $\entropy(\langle n/2,n/2\rangle) = n$.
However, a much better partition is
$\{(1,n),(2,3,\dots,n{-}1)\}$, of size $2$ and entropy
$\entropy(\langle 2,n-2\rangle) = \bigo(\lg n)$.

We doubt that the SUS partition minimizing $\entropy(\vSUS)$ can be
found within time $\bigo(n(1+\entropy(\vSUS)))$ or even
$\bigo(n(1+\lg\nSUS))$. Proving this right or wrong is an open challenge.

\subsection{Shuffled Monotone Sequences}
\label{sec:shuffl-monot-sequ}

No efficient algorithm is known to compute the minimum number $\nSMS$
of shuffled monotone sequences composing a permutation, let alone
finding a partition minimizing the entropy $\entropy(\vSMS)$ of the
lengths of the subsequences. The problem is NP-hard, by reduction to
the computation of the ``cochromatic'' number of the graph
corresponding to the
permutation~\cite{partitioningPermutationsIntoIncreasingAndDecreasingSubsequences}.

Yet, should such a partition into monotone subsequences be available,
and be of smaller entropy than the partitions considered in the
previous sections, this would yield an improved encoding
by doing just as in Theorem~\ref{thm:SUS} for SUS.

Note that 
it takes a difference by a superpolynomial margin between the values of
$\nSUS$ and $\nSMS$ to yield a noticeable difference between $\lg\nSUS$
and $\lg\nSMS$, and hence between the values of
$\entropy(\vSUS)$ and $\entropy(\vSMS)$. It seems unlikely that such a
difference would justify the difference of computing time between the
two types of partitions, also different by a superpolynomial margin to
the best of current knowledge (i.e., if $P\neq NP$).

\section{Conclusions}
\label{sec:conclusion}

\paragraph*{Relation between space and time}
\label{sec:relation-space-time}

Bentley and Yao~\cite{anAlmostOptimalAlgorithmForUnboundedSearching}
introduced a family of search algorithms adaptive to the
position of the element sought (also known as the ``unbounded search''
problem) through the definition of a family of adaptive codes
for unbounded integers, hence proving that the link between algorithms
and encodings was not limited to the complexity lower bounds suggested
by information theory.
Such a relation between ``time'' and ``space'' can be found in other
contexts: algorithms to merge two sets define an encoding for
sets~\cite{mergeSourceCoding}, and the binary results of
the comparisons of any deterministic sorting algorithm in the
comparison model yields an encoding of the permutation being sorted.

We have shown that some concepts originally defined for adaptive
variants of the algorithm MergeSort, such as runs and
shuffled sequences, are useful in terms of the compression of
permutations, and conversely, that concepts originally defined for
data compression, such as the entropy of the sets of run lengths,
are a useful addition to the set of difficulty measures previously
considered in the study of adaptive sorting algorithms.

Much more work is required to explore the application to the
compression of permutations and strings of the many other measures of
preorder introduced in the study of adaptive sorting algorithms.
Figure~\ref{fig:partialOrder} represents graphically some of those
measures of presortedness (adding to those described by 
Moffat and Petersson~\cite{anOverviewOfAdaptiveSorting}, those described in this and 
other recent work \cite{lrmTrees}) and a preorder on them based on optimality
implication in terms of the number of comparison performed. This is
relevant for the space of the corresponding permutation encodings, and
for the space used by the potential corresponding compressed
data structures for permutations.
Note that the reductions in this graph do not represent reductions in
terms of optimality of the running time to find the partitions.
For instance, we saw that $\HSMS$-optimality implies
$\HSUS$-optimality in terms of the number of comparison performed, but
not in terms of the running time.
In terms of data structures, this relates to the construction time of
the compressed data structure (as opposed to the space it takes).

\begin{figure}
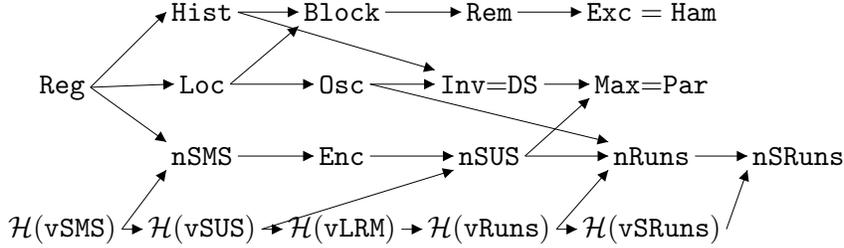

  \centering
  \renewcommand{\nBlock}{\idtt{Block}}
  \partialOrderOnPreorderMeasures
  \caption{Partial order on some measures of disorder for adaptive
    sorting. New results are on the bottom line.
  \label{fig:partialOrder}
}
\end{figure}

\paragraph*{Adaptive operators}
\label{sec:adaptive-operators}

It is worth noticing that, in many cases, the time to support the
operators on the compressed permutations is \emph{smaller} as the
permutation is more compressed, in opposition with the traditional
setting where one needs to decompress part or all of the data in
order to support the operators. 
This behavior, incidental in our study, is a very strong incentive to
further develop the study of difficulty or compressibility measures:
measures such that ``easy'' instances can both be compressed and
manipulated in better time capture the essence of the data.

\paragraph*{Compressed indices}
\label{sec:compressed-indices}

Interestingly enough, our encoding techniques for permutations
compress both the permutation and its index (i.e., the extra data to
speed up the operators).
This is opposed to previous
work~\cite{succinctRepresentationsOfPermutations} on the encoding of
permutations, whose index size varied with the size of the cycles of
the permutation, but whose data encoding was fixed; and to previous
work~\cite{succinctIndexesForStringsBinaryRelationsAndMultiLabeledTrees}
where the data itself can be compressed but not the index, to the
point where the space used by the index dominates that used by the
data itself.
This direction of research is promising, as in practice it is more
interesting to compress the whole succinct data structure or at least
its index, rather than just the data.

\paragraph*{Applications}
\label{sec:applications-1}

Permutations are
everywhere, so that compressing their representation helps compress
many other forms of data, and supporting in reasonable time the
operators on permutations yield support for other operators.

As a first example, consider a natural language text tokenized into word
identifiers.
Its {\em word-based inverted index} stores for each distinct word the list
of its occurrences in the tokenized text, in increasing order. This is a popular data
structure for text indexing \cite{BYRN11,WMB99}. By regarding the 
concatenation of the lists of occurrences of all the words, a permutation 
$\pi$ is obtained that is formed by $\nu$ contiguous ascending runs, where 
$\nu$ is the vocabulary size of the text. The lengths of those runs 
corresponds to the frequencies of the words in the text. Therefore our
representation achieves the zero-order word-based entropy of the text, which
in practice compresses the text to about 25\% of its original size
\cite{BCW90}. With $\pi(i)$ we can access any position of any inverted list,
and with $\pi^{-1}(j)$ we can find the word that is at any text position $j$.
Thus the representation contains the text and its inverted index within the
space of the compressed text.

A second example is given by compressed suffix arrays (CSAs), which are data 
structures for indexing general texts. A family of CSAs builds on a function
called $\Psi$ \cite{GV06,Sad03,GGV03}, which is actually a permutation. Much effort
was spent in compressing $\Psi$ to the zero- or higher-order entropy of the
text while supporting direct access to it. It turns out that $\Psi$ contains
$\sigma$ contiguous increasing runs, where $\sigma$ is the alphabet size of 
the text, and that the run lengths correspond to the symbol frequencies. Thus
our representation of $\Psi$ would reach the zero-order entropy of the text. It supports
not only access to $\Psi$ but also to its inverse $\Psi^{-1}$, which enables
so-called bidirectional indexes \cite{RNOM09}, which have several interesting
properties.  Furthermore, $\Psi$ contains
a number of strict ascending runs that depends on the high-order entropy of
the text, and this allows compressing it further \cite{NM07}.

From a practical point of view, our encoding schemes are simple enough
to be implemented. Some preliminary results on inverted indexes and
compressed suffix arrays show good performances on practical data sets.
As an external test, the techniques were successfully used to handle
scalability problems in MPI applications \cite{KMW10}.

\paragraph*{Followup}

Our preliminary results~\cite{BN09} have stimulated further research. This
is just a glimpse of the work that lies ahead on this topic.

While developing, with J. Fischer, compressed
indexes for Range Minimum Query indexes based on Left-to-Right Minima
(LRM) trees~\cite{Fis10,SN10}, we realized that LRM trees yield a
technique to rearrange in linear time $\nRuns$ contiguous ascending
runs of lengths forming vector $\vRuns$, into a partition of
$\nLRM=\nRuns$ ascending subsequences of lengths forming a new vector
$\vLRM$, of smaller entropy $\entropy(\vLRM)\leq\entropy(\vRuns)$
\cite{lrmTrees}.
Compared to a SUS partition, the LRM partition can have larger
entropy, but it is much cheaper to compute and encode. 
We represent it on
Figure~\ref{fig:partialOrder} between $\entropy(\vRuns)$ and
$\entropy(\vSUS)$.

While developing, with T. Gagie and Y.
Nekrich, an elegant combination of previously known compressed string
data structures to attain superior space/time trade-offs%
~\cite{alphabetPartitioning}, we realized that this
yields various compressed data structures for permutations $\pi$ such
that the times for $\pi()$ and $\pi^{-1}()$ are
improved to log-logarithmic. While those results subsume our
initial findings \cite{BN09}, the improved results now presented in
Theorem~\ref{thm:main2} are incomparable, and in particular superior
when the number of runs is polylogarithmic in $n$.

\section*{Acknowledgements}
  
  We thank Ian Munro, Ola Petersson and Alistair
  Moffat for interesting discussions.

\bibliographystyle{alpha}
  \bibliography{Biblio/general,Biblio/adaptiveAlgorithms,Biblio/hanoi,Biblio/perso,Biblio/aleh,Biblio/hci,Biblio/relationalDatabases,Biblio/biology,Biblio/informationRetrieval,Biblio/segmentIntersections,Biblio/cacheOblivious,Biblio/instanceOptimality,Biblio/spam,Biblio/chen,Biblio/interpolationSearch,Biblio/succinctDataStructures,Biblio/competitiveAnalysis,Biblio/intersection,Biblio/succinctPracticalEncodings,Biblio/computationalGeometry,Biblio/lca,Biblio/succinctTrees,Biblio/convexHull,Biblio/mechanismDesign,Biblio/unboundedSearch,Biblio/minimax,Biblio/unpublished,Biblio/google,Biblio/outputSensitiveAlgorithms,Biblio/xml,Biblio/gpgpu,Biblio/parameterizedComplexity,Biblio/gonzalo,paper}

\end{document}